\documentclass[submission,copyright,creativecommons]{eptcs}
\providecommand{\event}{Event Name} 

\usepackage{iftex}

\ifpdf
  \usepackage{underscore}         
  \usepackage[T1]{fontenc}        
  \usepackage[utf8]{inputenc}
\else
  \usepackage{breakurl}           
\fi

\usepackage{amsmath,amsfonts,amssymb,mathtools}
\usepackage{graphicx}
\usepackage{multirow}
\usepackage{subcaption}

\usepackage[linesnumbered,ruled,vlined]{algorithm2e}
\usepackage[english]{babel}

\usepackage[noend]{algorithmic}
\usepackage{tabularx}
\usepackage[table]{xcolor}
\usepackage{booktabs}

\usepackage{microtype}

\newcommand{\complClFont}[1]{\mathrm{#1}}
\newcommand{\problemFont}[1]{\textsc{#1}}
\newcommand{\mathCommandFont}[1]{\mathrm{#1}}

\newif\iflong
\newif\ifshort

\iflong
\else
\shorttrue
\fi

\RequirePackage{amsthm}
\usepackage{url}

\newtheorem{definition}{Definition}[section]
\newtheorem{example}[definition]{Example}

\let\phi\varphi

\DeclareMathOperator\var{var}

\newcommand{\PMABD}{\protect\ensuremath\problemFont{ABD}}

\newcommand{\wsat}{\protect\ensuremath\problemFont{WSAT}(2\mathrm{-CNF}^-)}
\newcommand{\ABD}{\protect\ensuremath\problemFont{ABD}}
\newcommand{\pABD}{\protect\ensuremath\problemFont{p-ABD}}
\newcommand{\RepABD}{\protect\ensuremath\problemFont{RepABD}}
\newcommand{\RepABDsub}{\protect\ensuremath\problemFont{RepABD}_\subseteq}
\newcommand{\pRepABD}{\protect\ensuremath\problemFont{p-RepABD}}

\newcommand{\SAT}{\protect\ensuremath\problemFont{SAT}}

\newcommand{\KB}{\protect\ensuremath\text{KB}}

\newcommand{\cclone}[1]{\ensuremath{\langle #1 \rangle}}

\newcommand{\str}[1]{\ensuremath{\mathrm{Str}}}

\newcommand{\clos}[1]{\left\langle #1 \right\rangle}
\newcommand{\closneq}[1]{\left\langle #1 \right\rangle_{\neq}}

\newcommand{\cloneFont}[1]{\mathsf{#1}}
\newcommand{\BR}{\protect\ensuremath{\cloneFont{BR}}}
\newcommand{\II}{\protect\ensuremath{\cloneFont{II}}}
\newcommand{\IL}{\protect\ensuremath{\cloneFont{IL}}}
\newcommand{\IS}[2]{\protect\ensuremath{\cloneFont{IS}^{#1}_{#2}}}
\newcommand{\ID}{\protect\ensuremath{\cloneFont{ID}}}
\newcommand{\IN}{\protect\ensuremath{\cloneFont{IN}}}
\newcommand{\IV}{\protect\ensuremath{\cloneFont{IV}}}
\newcommand{\IE}{\protect\ensuremath{\cloneFont{IE}}}
\newcommand{\IM}{\protect\ensuremath{\cloneFont{IM}}}
\newcommand{\IR}{\protect\ensuremath{\cloneFont{IR}}}
\newcommand{\IBF}{\protect\ensuremath{\cloneFont{IBF}}}

\newcommand{\epos}{\protect\ensuremath{\cloneFont{EP}}}
\newcommand{\eneg}{\protect\ensuremath{\cloneFont{EN}}}
\newcommand{\pos}{\protect\ensuremath{\cloneFont{P}}}
\newcommand{\negg}{\protect\ensuremath{\cloneFont{N}}}
\newcommand{\Card}[1]{\ensuremath{|\,#1\,|}}

\usepackage{todonotes}
\usepackage{tcolorbox} 
  
 \newcommand{\problemDef}[3]
{%
    \begin{tcolorbox}[arc=0.1mm,boxsep=-0.6mm,left=1.9mm,right=1.9mm,bottom=1.4mm,top=1.4mm,adjusted title={\strut \sc#1},colback=white!5]
    \noindent\textbf{Given:} #2.
    
    \noindent\textbf{Task:} #3
    \end{tcolorbox}
}

\usepackage{xspace}

\newcommand{\para}{\protect\ensuremath{\complClFont{para}}} 
\newcommand{\FPT}{\protect\ensuremath{\complClFont{FPT}}} 
\newcommand{\NP}{\protect\ensuremath{\complClFont{NP}}\xspace}

\newcommand{\DP}{\protect\ensuremath{\complClFont{DP}}} 
 
\newcommand{\Ptime}{{\protect\ensuremath{\complClFont{P}}}\xspace}

\newcommand{\co}{\protect\ensuremath{\complClFont{co}}}

\newcommand{\AllExpl}{\ensuremath{\mathcal{E}}}
\newcommand{\MinExpl}{\ensuremath{\mathcal{E}_M}}

\renewcommand{\mid}{\;|\;}
\newcommand{\eqdef}{\coloneqq}

\let\oldpar\paragraph
\renewcommand{\paragraph}[1]{\oldpar{#1{}.}}

\SetKwInput{KwData}{In}
\SetKwInput{KwResult}{Out}
\SetAlFnt{\small}
\SetAlCapFnt{\small}
\SetAlCapNameFnt{\small}
\SetAlCapHSkip{0pt}
\SetEndCharOfAlgoLine{}

\makeatletter
\newcommand{\algorithmfootnote}[2][\footnotesize]{
  \let\old@algocf@finish\@algocf@finish
  \def\@algocf@finish{\old@algocf@finish
    \leavevmode\rlap{\begin{minipage}{\linewidth}
    #1#2
    \end{minipage}}
  }
}
\makeatother

\usepackage{thmtools}
\usepackage{thm-restate}
\usepackage{apptools}

\declaretheorem[name=Lemma,sibling=definition]{lemma}
\declaretheorem[name=Theorem,sibling=definition]{theorem}

\title{Representative Sets in Propositional Abduction\thanks{Author names are stated in reverse alphabetical order.}}

\author{Johannes Schmidt
\institute{Jönköping University}
\and
Mohamed Maizia
\institute{Jönköping University \\ Linköping University}
\and
Victor Lagerkvist
\institute{Linköping University}
\and
Johannes K. Fichte
\institute{Linköping University}
}

\def\titlerunning{Representative Sets in Propositional Abduction}
\def\authorrunning{J. Schmidt, M. Maizia, V. Lagerkvist \& J.K. Fichte}

\begin{document}
\maketitle

\begin{abstract}
The propositional abduction problem is a well-known form of non-monotonic reasoning where we are asked to find an explanation of a given manifestation. Recently, there has been an influx of results asking more refined questions about the solution space rather than only individual solutions. For example, we might be interested in finding two solutions that are sufficiently far from each other (\emph{diverse} solutions) in the solution space. In this paper we consider a related \emph{representation} question where we ask if a given set of explanations $S$ can represent any other explanation (that is, whether their symmetric difference is smaller than a given $k$). We first study this problem from a classical complexity perspective and obtain a complete classification. While only a handful of cases are tractable, the increase in complexity compared to classical abduction is often smaller than expected. We then study the parameterized complexity for several parameters and obtain new tractable and hard cases. Interestingly, a full parameterized complexity classification would require resolving the parameterized complexity of the \emph{covering radius} problem from coding theory. To the best of our knowledge, no useful relationship between coding theory and non-monotonic reasoning has previously been established, but such connections seemingly become important when asking more complex questions about solution spaces.

\vspace{1em}
\noindent\textbf{Keywords:} Propositional Abduction, Computational Complexity, Post's Framework, Fine-grained Reasoning
\end{abstract}

\section{Introduction} \label{sec:intro}
\ifshort
The \emph{propositional abduction} problem is a well-known form of non-monotonic reasoning with many applications in e.g,\ AI and knowledge representation~\cite{Dellsen_2024,Wang-ZhouEtAl19,IgnatievNarodytskaMarques19,YuEtAl2023,HuEtAl2025}.
\fi
Here, we are asked to explain a given \emph{manifestation}. The explanation thus needs to logically entail the manifestation, and, to avoid trivial explanations, be consistent with the given knowledge base. For example, consider a medical diagnosis setting where a patient may develop a certain symptom depending on underlying conditions.
We could then have propositional variables $a$ (the patient has a weakened immune system), $b$ (the patient has a bacterial infection), $c$ (the patient has a viral infection), $d$ (the patient is exposed to severe environmental stress), and $m$ (the patient develops a high fever).
The knowledge base could then encode the following medical rules:

\begin{itemize}
    \item If the patient has a weakened immune system and a bacterial infection, then they develop a high fever.
    \item If the patient has a weakened immune system and a viral infection, then they develop a high fever.
    \item If the patient has a weakened immune system and is exposed to severe environmental stress, then they develop a high fever.
    \item The patient cannot simultaneously have a bacterial and a viral infection.
\end{itemize}

Formally, we could represent this as $\KB = \{ %
        a \land b \rightarrow m,\;
        a \land c \rightarrow m,\; %
        a \land d \rightarrow m,\;
        \neg (b \land c)\}%
        $.
The manifestation is the observation that the patient has developed a high fever, i.e., $M = \{m\}$, and the set of hypotheses is $H = \{a,b,c,d\}$. Then, for example, $\{a,b\}$ (weakened immune system and bacterial infection) and $\{a,b,d\}$ (additionally, stress) are both possible explanations, but, unless there is further evidence, one may argue that $\{a,b\}$ is preferable to $\{a,b,d\}$ since it makes fewer assumptions. In this scenario, it seems desirable to consider \emph{all} minimal explanations, revealing that the fever can be explained by a weakened immune system combined with exactly one of the mutually exclusive infections, or with environmental stress.

As might be expected, computing/counting all (minimal) explanations is computationally expensive~\cite{HermannPichler10,CreignouEtAl19}, and even  deciding existence of just a single explanation is $\Sigma^P_2$-complete~\cite{EiterGotlob95}. 
Nevertheless, there has been many attempts to reason about the set of solutions, e.g.\ by identifying \emph{facets}~\cite{alrabbaa-et-al-rulemlpr2018,lagerkvist2025a,RusovacEtAl24}, and finding \emph{diverse} solutions (induced by a given distance metric between solutions). For example, diverse solutions have been considered for \emph{answer set programming}~\cite{eiter2013finding}, abduction~\cite{schmidt2025}, \emph{constraint satisfaction problems}~\cite{Hebrard2005}, satisfiability problems~\cite{MisraM024}, and a wealth of graph problems~\cite{Baste2020,Fomin2020,Fomin2021}. 

Inspired by the success of this approach and recent work on answer set programming~\cite{BohlGagglRusovac23}  we in this paper consider a related question: does there exist a set of explanations that \emph{represent} all (minimal) explanations? 
We formulate this as a decision problem and then qualify our problem with a parameter $k \geq 0$, a set of explanations $S$, and want to know if every explanation is within distance $k$ from at least one explanation in $S$.
In this case $S$ is called \emph{(k-)representative}.
Thus, unless $k$ is large, we expect the set $S$ to correlate with diverse solutions. 

We denote this problem by $\RepABD$, and the corresponding problem for representing subset minimal explanations, by $\RepABDsub$. If no assumptions are imposed on the knowledge base $\KB$ it is easy to show that both these problems are $\Pi^P_2$-complete, and we therefore attempt a more fine-grained picture of the complexity with restricted knowledge bases (e.g., whether it is in Horn, or in $2\text{-}CNF$). We write $\RepABD(\Gamma)$ ($\RepABDsub(\Gamma)$) for this problem where $\Gamma$ is a set of relations, and then require that the knowledge base is given by a conjunctive $\Gamma$-formula, i.e., each atom is of the form $R(x_1, \ldots, x_k)$ for $R \in \Gamma$ and variables $x_1, \ldots, x_k$. We formally introduce this problem in Section~\ref{sec:repr} together with a few useful definability notions. Then, for our main technical contributions, we (in Section~\ref{sec:classical}) classify the classical complexity of $\RepABD(\Gamma)$ and $\RepABDsub(\Gamma)$ for all possible choices of $\Gamma$. 

Our classification reveals that the two problems have few tractable cases. For example, if $\Gamma$ can express the ``trivial'' unary relation $\mathsf{t} = \{(0), (1)\}$ then $\RepABD(\Gamma)$ is at least coNP-hard. Surprisingly, $\RepABDsub(\Gamma)$ fares marginally better in comparison and we prove that it is in P if each relation is \emph{strictly essentially positive}, or the dual case of being \emph{strictly essentially negative}. However, it should be noted that tractability in this case stems from rather trivial reasons, and to extend the tractable fragments we (in Section~\ref{sec:para}) turn to \emph{parameterized complexity}. Here, we relax polynomial time to additionally allow a factor $f(p)$ where $p \in \mathbb{N}$ is a parameter (e.g., $|H|$, $|M|$, $k$, or a graph parameter of $\KB$) and $f \colon \mathbb{N} \to \mathbb{N}$ a computable function. A problem admitting such a running time is said to be \emph{fixed-parameter tractable} (FPT). The parameterized complexity of abduction is well understood for many natural parameters~\cite{MahmoodEtAl21,FellowsEtAl12} and admits non-trivial FPT cases, so there is reason for a certain optimism. We consider many different parameters ($k$, $|H|$, $|M|$, and $|S|$) and  establish FPT for $|H|$ if $\Gamma$ is \emph{Schaefer} and for $|S|$ if $\Gamma$ is strictly essentially positive. Importantly, we complement this with many lower bounds that rule out FPT under widely believed conjectures in parameterized complexity.

Interestingly, by considering such ``fine-grained'' questions about the set of explanations, we are able to make connections to problems previously unconnected to non-monotonic reasoning. For example, one of our main sources of hardness stems from the \emph{covering radius} problem~\cite{frances1997covering}, and one of our main FPT results are based on a reduction to the \emph{closest string problem}~\cite{GrammNR03}. 
As we show, a complete parameterized complexity classification of $\RepABD(\Gamma)$ would simultaneously need to resolve the parameterized complexity of the covering radius problem (with parameter $r$). We discuss this and other questions in Section~\ref{sec:conclusions}.

\ifshort
Due to space constraints the proof of statements marked with $\star$ have been omitted. 
\fi

\section{Preliminaries} \label{sec:preliminaries}
%

\newcommand{\BigO}[1]{\ensuremath{\mathcal{O}(#1)}}
\newcommand{\CCard}[1]{\ensuremath{||#1||}}
\newcommand{\preduction}{\leq^{\mathCommandFont{\Ptime}}_m}

We follow standard notions in 
computational complexity theory~\cite{downey2013fundamentals}, 
%
and propositional logic.
Below, we briefly state the most important notions.

\subsection{Computational Complexity}
%
%

\ifshort
Let $\Sigma$ and $\Sigma'$ be finite alphabets. An \emph{instance} is a string $I \in \Sigma^*$ and $\CCard{I}$ denotes its size. A \emph{decision problem} is a language $L \subseteq \Sigma^*$. Recall that $\Ptime{}$ and $\NP$ are the classes of deterministically and non-deterministically polynomial-time solvable decision problems~\cite{Cook71}. A polynomial-time many-to-one reduction~($\preduction$) from $L$ to $L'$ is a function $r : \Sigma^* \rightarrow {\Sigma'}^*$ such that $I \in L$ if and only if $r(I) \in L'$ and $r$ is computable in time $\BigO{\CCard{I}^c}$ for some constant~$c$. We also use the Polynomial Hierarchy (PH)
where $\Delta^\Ptime_0 \eqdef \Pi^\Ptime_0 \eqdef \Sigma^\Ptime_0 \eqdef \Ptime$, $\Delta^\Ptime_{i+1} \eqdef P^{\Sigma^\Ptime_i}$, $\Sigma^\Ptime_{i+1} \eqdef \NP^{\Sigma^\Ptime_i}$, and $\Pi^\Ptime_{i+1} \eqdef \text{co}\NP^{\Sigma^\Ptime_i}$ for $i>0$. For a decision problem $X$, we write $\overline{X}$ for its complement.

A \emph{parameterized problem} is a set $L \subseteq \Sigma^* \times \mathbb{N}$. It is \emph{fixed-parameter tractable} ($\mathrm{FPT}$) if membership of $(I,k)$ can be decided in time $f(k)\cdot \CCard{I}^c$ for a computable function $f$ and a constant~$c$. A \emph{fixed-parameter ($\mathrm{FPT}$) reduction} from $L$ to $L'$ is a mapping $P$ (with respect to computable $f,p : \mathbb{N} \to \mathbb{N}$
and a constant $c$) 
where we for any
$(I,k) \in \Sigma^* \times \mathbb{N}$) have
(1) $(I,k) \in L$ if and only if $P(I,k) = (I',k') \in L'$, (2) $k' \leq p(k)$, and (3) $P(I,k)$ can be computed in $f(k) \cdot \CCard{I}^c$ time.
A problem is \emph{$\mathrm{W}[1]$-hard} if it admits an $\mathrm{FPT}$-reduction from \textsc{Independent Set}. Another important $\mathrm{W}[1]$-hard problem is \emph{weighted SAT} where we ask if a given propositional formula $\phi$ has a model of weight $\geq k$. We let $\wsat$ denote its restriction to negative $2\text{-}CNF$. Finally, for a classical complexity class $C$, we define $\para\text{-}C$ as the class of all parameterized problems reducible to a language in $C$ via a computable parameter transformation, allowing us to speak of classes such as $\para\text{-}\NP$, $\para\text{-}\co\NP$, $\para\text{-}\DP$, and $\para\text{-}\co\DP$.
\fi

\paragraph{Propositional Logic}
A \emph{literal} is a variable $x$ or its negation $\neg x$. 
A \emph{clause} is a disjunction of literals, often represented as a set. A clause of arity 1, i.e., either $(x)$ or $(\neg x)$, is a {\em unit clause}.
We work in a general setting where atoms can be expressions of the form $R(x_1,  \ldots, x_r)$ for variables $x_1, \ldots, x_r$ and an $r$-ary relation $R \subseteq \{0,1\}^r$. A function $f \colon \{x_1, \ldots, x_r\} \to \{0,1\}$ is then said to satisfy an atom $R(x_1, \ldots, x_r)$ if $(f(x_1), \ldots, f(x_r)) \in R$. A (conjunctive) {\em propositional formula} $\varphi$ is a conjunction of atoms and we
write $\var(\varphi)$ for its set of variables.
A mapping $\sigma\colon \var(\varphi) \mapsto \{0,1\}$ is called an \emph{assignment} to the variables of~$\varphi$ and
a \emph{model} of a formula $\varphi$ is an assignment to $\var(\varphi)$ that satisfies $\varphi$.
For two formulas $\psi$ and $\varphi$, we write $\psi  \models \varphi$ if every model of $\psi$ also satisfies $\varphi$.

\subsection{Restrictions of Constraint Languages}

We work in a generalized setting where atoms can be formed by combining relations and variables. Then, a {\em constraint language} $\Gamma$ is a set of Boolean relations, and a {\em $\Gamma$-formula} over some variables $V$ is a propositional formula $\varphi$ where $R \in \Gamma$ and $x_1, \ldots, x_r \in V$ for each atom in the formula $R(x_1, \ldots, x_r)$. For a constraint language $\Gamma$, we write $\SAT(\Gamma)$ for the problem of deciding if a given $\Gamma$-formula admits at least one model.
%
%
%
%
Usually, we do not 
distinguish between the relation or a clause defining the relation. For example, we simply write $(x)$ for the unary relation $\{(1)\}$, $(\neg x)$ for $\{(0)\}$, $(x_1 \rightarrow x_2)$ or $(\neg x_1 \lor x_2)$ for $\{(0,0), (0,1), (1,1)\}$, and so on. The empty set $\emptyset$ is the (nullary) relation that is always false, we write $R_=$ for the equality relation $\{(0,0), (1,1)\}$ (but often written in infix form as $(x = y)$ instead of $R_=(x, y)$), and $\mathsf{t}$ for the trivial unary relation $\{(0), (1)\}$ that is always true.

\begin{table*}
{\scriptsize
  \centering
  \rowcolors{2}{gray!25}{white}
  \resizebox{\linewidth}{!}{%
    \begin{tabular}{lll}\toprule
      co-clone        & clauses/equation                                                                                            & name/indication                             \\\midrule
      $\BR$ ($\II_2$) & all clauses                                                                                            & all Boolean relations                       \\
      $\II_0$         & at least one negative literal per clause                                                               & 0-valid                                     \\
      %
      %
      $\IN_2$         & NAE = $\{0,1\}^3 \setminus \{000,111\}$                                                                                    & complementive                               \\ 
      $\IN$           & DUP = $\{0,1\}^3 \setminus \{101, 010\}$                                                                                    & complementive and 1- and 0-valid            \\
      $\IE_2$         & clauses with at most one positive literal                                                              & Horn                                        \\
      %
      %
      $\IV_2$         & clauses with at most one negative literal                                                              & dualHorn                                    \\
      %
      %
$\IL_2$               & all affine clauses (all linear equations)                                                              & affine                                      \\
$\IL$                 & $(x_1 \oplus \dots \oplus x_n = 0)$, $n$ even                                                          & affine and 1- and 0-valid                   \\
$\ID_2$               & clauses of size 1 or 2                                                                                 & Krom, bijunctive, $2\text{-}CNF$                      \\
$\ID$                 & affine clauses of size 2                                                                               & strict 2-affine                             \\
%
%
$\IM$                 & $(x_1 \rightarrow x_2)$                                                                                & implicative and 1- and 0-valid              \\
%
%
$\IS{}{12}$           & $(x_1), (\neg x_1 \lor \dots \lor \neg x_n), n \geq 0, (x_1 = x_2)$                                    & essentially negative ($\eneg$)              \\
$\IS{k}{11}$          & $(x_1 \rightarrow x_2), (\neg x_1 \lor \dots \lor \neg x_n), k \geq n \geq 0$                          & -                                           \\
$\IS{k}{1}$           & $(\neg x_1 \lor \dots \lor \neg x_n), k \geq n \geq 0, (x_1 = x_2)$                                    & negative ($\negg$) of width $k$             \\

$\IS{}{02}$           & $(\neg x_1), (x_1 \lor \dots \lor x_n), n \geq 0, (x_1 = x_2)$                                         & essentially positive ($\epos$)              \\

$\IR_1$               & $(x_1), (x_1 = x_2)$                                                                                   & -                                           \\
$\IR_0$               & $(\neg x_1), (x_1 = x_2)$                                                                              & -                                           \\
$\IR$ ($\IBF$)        & $(x_1 = x_2)$                                                                                          & -                                           \\\bottomrule
\end{tabular}}	
\caption{Overview of some co-clones and clause descriptions \cite{NordhZ08}.}%
\label{tab:small_bases}
}
\end{table*}

For a constraint language $\Gamma$ and $k \geq 1$, we often use the notation $k$-$\Gamma$ for the set of 
relations/clauses of arity at most $k$ (e.g., $2\text{-}CNF$ contains all clauses of arity 1 and 2). 
Additionally, for a language $\Gamma$ we, let (1) 
$\Gamma^- = \Gamma \setminus \{(x), (\neg x)\}$ be $\Gamma$ without the two unit clauses, and (2) 
$\Gamma^+ = \Gamma \cup \{(x), (\neg x)\}$ be $\Gamma$ expanded with the two unit clauses. 
A language $\Gamma$ 
is {\em $b$-valid} for $b \in \{0,1\}$, if $(b, \ldots, b)\in R$ 
for 
%
each $R \in \Gamma$.
%
We introduce the most important constraint languages for 
the purpose of this paper in Table~\ref{tab:small_bases}. We also write $\text{EP}^-$ for the set of \emph{strictly essentially positive clauses} $\text{EP} \setminus \{(x = y)\}$ and $\text{EN}^- = \text{EN} \setminus \{(x = y)\}$ for the set of \emph{strictly essentially negative clauses}. 
To avoid doing an
exhaustive case analysis of {\em all} possible constraint languages we introduce a useful closure property on relations. 
Say that an $r$-ary 
relation $R$ has a \emph{primitive positive definition} (pp-definition) over  
$\Gamma$ if \\[-1.0em]
$$R(x_1, \dots, x_r) := \exists y_1, \dots, y_n\, .\, \varphi(x_1, \ldots, x_r, y_1, \ldots, y_n)$$ 
\noindent
where $\varphi$ is a $(\Gamma \cup \{\emptyset, \mathsf{t}, (x = y)\})$-formula. Thus, put otherwise, $R$ can be defined as the set of models of $\exists y_1, \ldots, y_n\, .\, \varphi(x_1, \ldots, x_r, y_1, \ldots, y_n)$ with respect to the free variables $x_1, \ldots, x_r$. 


\begin{definition}
For a constraint language $\Gamma$ we let $\cclone{\Gamma}$ be the smallest set of relations containing $\Gamma$ and where $R \in \Gamma$ for any pp-definable relation $R$ over $\Gamma$.
\end{definition}  
The set $\Gamma$ is in this context said to be a \emph{base}, and $\cclone{\Gamma}$ is sometimes called a \emph{relational clone}, or a \emph{co-clone}. 
For details, we refer to the work by~\cite{BohlerRSV05}. 
We note that the complexity of $\SAT(\Gamma)$ is completely determined due to Schaefer's famous dichotomy result: it is polynomial
if $\Gamma$ is 0- or 1-valid or \emph{Schaefer} (that is, $\Gamma$ is Horn, or dualHorn, or affine, or $2\text{-}CNF$) and $\NP$-complete otherwise \cite{Schaefer78}.

\subsection{Propositional Abduction}
%
%
%
%
%
Let $\Gamma$ be a constraint language, for example, a set of clauses.
An instance~$I$ of the \emph{positive propositional} abduction
problem over $\Gamma$, $\PMABD(\Gamma)$ for short, 
is a tuple~$I=(\KB, H, M)$ 
with $\KB$ being a $\Gamma$-formula over a finite set of Boolean variables called the \emph{knowledge base} (or \emph{theory}), $H \subseteq \var(\KB)$ called
\emph{hypotheses}, $M\subseteq \var(\KB)$ called
\emph{manifestations}. 
%
%
%
Since we have defined a $\Gamma$-formula as a conjunctive formula with atoms from $\Gamma$ we sometimes take the liberty of viewing the knowledge base as a set rather than as a formula.
%
%
%
%
%
%
A \emph{positive explanation}~$E$, \emph{explanation} for short, is a subset~$E\subseteq H$ such that (i)~$\KB \land E$ is satisfiable and (ii)~$\KB \land E \models M$.
%
%
%
%
%
An explanation~$E$ is \emph{(subset-)minimal} if no other set~$E' \subsetneq E$ is an explanation of~$I$.

The problem~$\PMABD(\Gamma)$ 
asks whether there is an explanation, which in the decision context is the same as asking whether there is a minimal explanation.
If $\Gamma$ is arbitrary, we omit $\Gamma$ from the problem and
write $\PMABD$. Note that the complexity of $\PMABD(\Gamma)$ is completely determined \cite{NordhZ08}.
%

%
We write $\AllExpl(I)$ to refer to the set of all explanations and $\MinExpl(I)$ for the set of all subset-minimal explanations.
%
%
%
%
%
%
%




    \newcommand{\lwed}{\textbf{\underline{w}ednesday}}
    \newcommand{\swed}{w}
    \newcommand{\lcalm}{\textbf{\underline{c}alm}}
    \newcommand{\scalm}{c}
    \newcommand{\lrace}{\textbf{\underline{n}o-race}}
    \newcommand{\srace}{n}
    \newcommand{\lstorm}{\textbf{\underline{s}torm}}
    \newcommand{\sstorm}{s}
    \newcommand{\lrain}{\textbf{\underline{r}aining}}
    \newcommand{\srain}{r}

     \begin{example}\label{ex:running}
        Consider the abduction example from Section~\ref{sec:intro}
        where  $I=(\KB, H,M)$
        with  
        %
        $\KB = \{ %
        a \land b \rightarrow m,\;
        a \land c \rightarrow m,\; %
        a \land d \rightarrow m,\;
        \neg (b \land c)\}%
        $, 
         manifestation~$m$, and the hypotheses~$\{a,b,c,d\}$. 
        Then 
        $\MinExpl(I) = \{\{a,b\}, \{a,c\},\{a,d\}\}$ and 
        $\AllExpl(I) = \{\{a,b\}, \{a,c\},\{a,d\}, \{a,b,d\}, \{a,c,d\}\}$.
\end{example}

\section{The Representative Explanation Problem} \label{sec:repr}

We begin the technical part of the paper by formally introducing our problem as well as the simplifying algebraic machinery.
\ifshort
Let $I=(\KB, H, M)$ be an $\PMABD$ instance. For a set $E \subseteq H$, define 
$S_k(E) = \{E_1 \in \AllExpl(I) \mid d(E,E_1) \leq k\}$
where $d(\cdot,\cdot)$ denotes the symmetric difference, i.e., 
$d(E_1, E_2) = \Card{E_1 \triangle E_2} = \Card{(E_1 \cup E_2) \setminus (E_1 \cap E_2)} 
= \Card{\{x \in H \mid x \in E_1 \text{ and } x \notin E_2 \text{ or } x \notin E_1 \text{ and } x \in E_2\}}$, i.e., $S_k(E)$ is the set of explanations within Hamming distance $k$ of $E$.
\fi
For a fixed $k$, an explanation $E \in \AllExpl(I)$ is called \emph{k-representative}, if
$|S_k(E)|$ is maximal. That is, such an $E$ represents maximally many explanations.
We extend this notion to sets of explanations as follows:
for a fixed $k$, a set of explanations $S \subseteq \AllExpl(I)$ is called \emph{$k$-representative}, if
\ifshort
$\AllExpl(I) = \bigcup_{E \in S} S_k(E).$
\fi
%

We then consider the following problem where the task is to verify if a set of explanations is $k$-representative or not.

\problemDef{$\RepABD(\Gamma)$}{An $\PMABD(\Gamma)$ instance $I = (\KB, H, M)$, a set $S \subseteq \AllExpl(I)$, and $k \geq 1$}{Is $S$ $k$-representative for $I$?}

We also consider $\RepABD(\Gamma)$ restricted to $\subseteq$-minimal explanations (given now $S \subseteq \MinExpl(I)$ check if $\MinExpl(I) = \bigcup_{E \in S} S_k^M(E)$, where $S_k^M(E) = \{E_1 \in \MinExpl(I) \mid d(E,E_1) \leq k\}$) and write $\RepABDsub(\Gamma)$ for this variant. 

\begin{example}
We continue with Example~\ref{ex:running} where we had    $\KB = \{ %
        a \land b \rightarrow m,\;
        a \land c \rightarrow m,\; %
        a \land d \rightarrow m,\;
        \neg (b \land c)\}%
        $,
        $M = \{m\}$, $H = \{a,b,c,d\}$, and      $\AllExpl(I) = \{\{a,b\}, \{a,c\},\{a,d\}, \{a,b,d\}, \{a,c,d\}\}$.
        Then $S_1 = \{\{a,b\}\}$ is not 2-representative, since $d(\{a,b\}, \{a,c,d\}) = 3$. But $S_2 = \{\{a,d\}\}$ is 2-representative, since $d(\{a,d\}, E) \leq 2$ for all $E \in \AllExpl(I)$, as is easily verified.
\end{example}


Before turning to the complexity of $\RepABD(\Gamma)$ and $\RepABDsub(\Gamma)$ we show how to apply the algebraic approach --- with a minor modification. First, we say that a pp-definition $R(x_1, \dots, x_r) := \exists y_1, \dots, y_n\, .\, \varphi(x_1, \ldots, x_r, y_1, \ldots, y_n)$ is \emph{equality-free} (efpp) if each atom in $\varphi$ is from $\Gamma \cup \{\emptyset\}$, i.e., we do not allow (1) the equality relation, or (2) the full relation $\mathsf{t}$. We let $\closneq{\Gamma}$ be the smallest set of relations containing $\Gamma$ closed under such definitions. We have the following basic characterization (where a relation $R$ is said to be constant if $|R| = 1$).

Before stating and proving the lemma we need a few additional preliminaries. For a tuple $t = (x_1, \ldots, x_i, \ldots, x_n) \in \{0,1\}^n$ and $i \in [n] = \{1, \ldots, n\}$ we write $t[i] = x_i$ for the $i$th component. For $R \subseteq \{0,1\}^n$ of arity $n$  we say that $i \in [n]$ is \emph{fictitious} if $(x_1, \ldots, x_{i-1}, x_i, x_{i+1}, \ldots, x_n) \in R$ if and only if we have $(x_1, \ldots, x_{i-1}, \bar{x}_i, x_{i+1}, \ldots, x_n) \in R$, and it is \emph{redundant} if there exists $j \in [n]$, $j \neq i$, such that  $t[i] = t[j]$ for any $t \in R$. Furthermore, say that $R$ is \emph{irredundant} if it has no redundant coordinates, and that it is \emph{prime} if it is irredundant and has no fictive coordinates. We can now relate these notions to efpp-definability as follows.

\begin{lemma}($\star$) \label{lemma:can_define}
    Let $\Gamma$ be a set of Boolean relations and let $R \in \clos{\Gamma}$. If
    \begin{enumerate}
        \item 
        $R$ is prime then $R \in \closneq{\Gamma}$,
        \item 
        $R$ is irredundant then $R \in \closneq{\Gamma \cup \{\mathsf{t}\}}$, and
        \item
        otherwise
        $R \in \closneq{\Gamma \cup \{R_=\}}$.
    \end{enumerate}
\end{lemma}

This in turn leads to the following classification of efpp-closed sets.
\begin{lemma} \label{lemma:weak_def}
Let $\Gamma$ be a set of Boolean relations. If
\begin{enumerate}
\item if $R_= \in \closneq{\Gamma}$ then $\mathsf{t} \in \closneq{\Gamma}$,
    \item if $R_= \in \closneq{\Gamma}$ then $\closneq{\Gamma} = \cclone{\Gamma}$, 
    \item 
    if $\Gamma$ contains a non-constant relation then 
    $\closneq{\Gamma} = \closneq{\Gamma \cup \{\mathsf{t}\}}$, and
     \item 
     if $\sf{t} \notin \closneq{\Gamma}$ then 
     $\closneq{\Gamma} \subsetneq \closneq{\Gamma \cup \{\mathsf{t}\}} \subsetneq \closneq{\Gamma \cup \{R_=\}} = \clos{\Gamma}$, \\and 
    $\closneq{\Delta} \in \{\closneq{\Gamma}, \closneq{\Gamma \cup \mathsf{t}}, \closneq{\Gamma \cup \{R_=\}}\}$ for any $\Delta$ such that $\clos{\Gamma} = \clos{\Delta}$.
\end{enumerate}
\end{lemma}

\begin{proof}
    The first claim follows immediately since we can define $\mathsf{t}(x)$ via the definition $(x = x)$. For the second claim: if we can efpp-define $R_=$ then we (by the first claim) can also efpp-define $\mathsf{t}$, and any pp-definition (possibly using $R_=$- or $\mathsf{t}$-constraints) can be converted into a suitable efpp-definition. For the third claim, let $R \in \Gamma$ be a non-constant relation, of arity, say $r$, and let $1 \leq i \leq r$ be an argument such that $|\{x_i \mid (x_1, \ldots, x_i, \ldots, x_r) \in R\}| = 2$. Then $\mathsf{t}(x) := \exists x_1, \ldots, x_{i-1}, x_{i+1}, \ldots, x_r . R(x_1, \ldots, x_{i_1}, x, x_{i+1}, \ldots, x_r)$.

For the fourth and last claim, $\closneq{\Gamma} \subsetneq \closneq{\Gamma \cup \{\mathsf{t}\}}$ follows from the assumption that $\mathsf{t} \notin \closneq{\Gamma}$. For the second inclusion $\closneq{\Gamma \cup \{\mathsf{t}\}} \subsetneq \closneq{\Gamma \cup \{R_=\}}\}$, observe that $R_=$ has no fictitious argument, and, hence, if $R_= \in \closneq{\Gamma \cup \{\mathsf{t}\}}$ then $R_= \in \closneq{\Gamma}$ (by Lemma~\ref{lemma:can_define}). But then (as established in the first item of this lemma) $\closneq{\Gamma} = \cclone{\Gamma}$ which contradicts the assumption that $\mathsf{t} \notin \cclone{\Gamma}$. Hence, $R_= \notin \closneq{\Gamma \cup \{\mathsf{t}\}}$ and the inclusion $\closneq{\Gamma \cup \{\mathsf{t}\}} \subsetneq \closneq{\Gamma \cup \{R_=\}}$ must be proper. Now, consider a $\Delta$ such that $\clos{\Delta} = \clos{\Gamma}$. The claim  that $\closneq{\Delta} \in \{\closneq{\Gamma}, \closneq{\Gamma \cup \mathsf{t}}, \closneq{\Gamma \cup \{R_=\}}\}$ then follows through a similar case analysis:  if every relation is prime, then $\closneq{\Delta} = \closneq{\Gamma}$, if every relation is irredundant then $\closneq{\Delta} \in \{\closneq{\Gamma}, \closneq{\Gamma \cup \{\mathsf{t}\}}\}$, and if these two cases do not apply then $\closneq{\Delta} = \clos{\Gamma}$.
\end{proof}

We remark that disallowing equality is a fairly standard assumption for certain problems~\cite{CreignouE014,MahmoodEtAl21, MaMeSc23, CrOlSc23} but $\mathsf{t}$ is normally such a harmless relation that it is not explicitly acknowledged in definitions. However, we will later see that their presence \emph{do} make a difference for the $\RepABD(\Gamma)$ problem in the sense that there are languages $\Gamma$ such that $\RepABD(\Gamma)$ is in P but $\RepABD(\Gamma \cup \{\mathsf{t}\})$ is intractable. With this in mind we obtain the following basic reducibility result.

\begin{lemma} ($\star$) \label{lem:baseIndneq}
    Let $\Gamma$ and $\Gamma'$ be two constraint languages. If $\Gamma' \subseteq \closneq{\Gamma}$, then
     $\RepABD(\Gamma') \preduction
     \RepABD(\Gamma)$ and $\RepABDsub(\Gamma') \preduction
     \RepABDsub(\Gamma)$.
\end{lemma}

We further need the following expressiveness result.
\begin{lemma}($\star$)\label{lem:nessp-nessn-equality}
    Let $\Gamma$ be a constraint language. If $\Gamma \not \subseteq \closneq{\mathrm{EP}^-}$ and $\Gamma \not \subseteq \closneq{\mathrm{EN}^-}$
    then $(x=y) \in \closneq{\Gamma}$ and $\clos{\Gamma} = \closneq{\Gamma}$.
\end{lemma}



\section{Classical Complexity} \label{sec:classical}

We begin by analyzing the ``classical'' complexity of $\RepABD(\Gamma)$, i.e., whether it is in P, or in an intractable class, and thus study the complexity of the problem up to polynomial-time reductions. We first observe a straightforward upper bound.

\begin{lemma}\label{lem:PiP2-membership}
$\RepABD(\Gamma)$ is in $\Pi^P_2$.
\end{lemma}
\begin{proof}


The following non-deterministic algorithm  shows that $\overline{\RepABD(\Gamma)}$ (recall that this is the complement of $\RepABD(\Gamma)$) is in $\Sigma^P_2$. We are given an instance $(\KB, H, M)$, a set $S \subseteq \AllExpl(I)$, and $k \geq 1$. 
We then guess an $E' \subseteq H$, and verify that \\ 
(1) $E'\in \AllExpl(I)$ --- feasible with an NP- and a coNP-oracle (first check if $\KB \land E'$ is satisfiable, and then if $\KB \land E' \models M$)\\
(2) $E' \notin \bigcup_{E \in S} S_k(E)$ --- feasible with a coNP-oracle (since the complement question is in NP: guess an $E\in S$ and verify that $d(E,E') \leq k$).
\end{proof}

Recall that for a language $\Gamma$ we let $\Gamma^+$ be $\Gamma$ expanded with the two constant Boolean relations. We have the following useful conditional upper bound.


\begin{lemma}($\star$) \label{lemma:sat_to_rep}
If $\SAT(\Gamma^+) \in \Ptime $ then $\RepABD(\Gamma) \in \co\NP$, for any constraint language $\Gamma$.
 \end{lemma}

The following lemma lets us inherit numerous hardness results from $\ABD$.

\begin{lemma}($\star$)\label{lem:ABD-to-RepABD}
Let $\Gamma$ be a constraint language. Then $\overline{\ABD(\Gamma)} \preduction \RepABD(\Gamma)$.
\end{lemma}

However, the following lemma proves that $\RepABD$ is generally harder than $\ABD$ since it establishes $\co\NP$-hardness for many fragments where $\ABD$ is tractable. Recall that $\mathsf{t} = \{(0),(1)\}$ and that $T = \{(1)\}$.
 In the reduction we use the NP-complete \emph{covering radius problem} \cite{frances1997covering}, where an instance is given by a binary code $C \subseteq \{0,1\}^n$ and an integer $r$, and the question is whether there is a vector $v \in \{0,1\}^n$ that has Hamming distance greater than $r$ from every element in $C$ (thus the \emph{covering radius} of $C$ is greater than $r$).

\begin{lemma}\label{lem:P1ABDcoNPhard}
    $\RepABD(\Gamma)$ is coNP-hard for any $\Gamma$ such that $\mathsf{t} \in \closneq{\Gamma}$ or $T \in \closneq{\Gamma}$.
\end{lemma}

\ifshort
\begin{proof}[Sketch]
We reduce from the covering radius problem. Intuitively, we encode an instance $(C,r)$ of this problem by introducing a fresh variable for each coordinate of the codewords and constructing $\KB$ so that all subsets of these variables form candidate explanations. The set $S$ is defined to correspond precisely to the codewords in $C$, and the distance parameter $k$ is set to $r$. 
Under this construction, any explanation that is \emph{not} $k$-represented by $S$ corresponds exactly to a vector at Hamming distance greater than $r$ from all codewords in $C$. Conversely, if every vector lies within distance $r$ of some codeword, then all explanations are $k$-represented, yielding a positive instance. 
\end{proof}
\fi

We obtain the following complexity classification of $\RepABD(\Gamma)$.

\begin{theorem} 
Let $\Gamma$ be a constraint language. Then $\RepABD(\Gamma)$ is
\begin{enumerate}
    \item $\Pi^P_2$-complete if $\IN_2 \subseteq \clos{\Gamma} \subseteq \II_2$ or $\II_0 \subseteq \clos{\Gamma} \subseteq \II_2$,
    \item $\co\NP$-hard and $\NP$-hard if $\IN \subseteq \clos{\Gamma}$,
    \item $\co\NP$-complete if $C \subseteq \clos{\Gamma} \subseteq D$ for $C \in \{\IS{2}{1}, \IM, \IR_1, \ID, \IL\}$ and $D \in \{\IE_2, \IL_2, \IV_2, \ID_2\}$,
    \item $\co\NP$-complete if $\IBF \subseteq \clos{\Gamma} \subseteq \IR_0$ and $\mathsf{t} \in \closneq{\Gamma}$, and
    \item $\in \Ptime$ otherwise.
\end{enumerate}
\end{theorem}
\begin{proof}
    First, assume that $\IN_2 \subseteq \clos{\Gamma} \subseteq \II_2$ or that $\II_0 \subseteq \clos{\Gamma} \subseteq \II_2$. From Lemma~\ref{lem:PiP2-membership} we know that $\RepABD(\Gamma) \in \Pi^P_2$, and $\ABD(\Gamma)$ is $\Sigma^P_2$-complete in this case, which in combination with Lemma~\ref{lem:ABD-to-RepABD} gives the desired $\Pi^P_2$-hardness. Second, assume that $\IN \subseteq \clos{\Gamma}$. Then $\ABD(\Gamma)$ is $\co\NP$-hard, and Lemma~\ref{lem:ABD-to-RepABD} therefore gives $\NP$-hardness for $\RepABD(\Gamma)$. For the $\co\NP$-hardness claim we first observe that $\Gamma$ must contain a non-constant relation, which from Lemma~\ref{lemma:weak_def} implies that $\mathsf{t} \in \closneq{\Gamma}$, and we finally get $\co\NP$-hardness from Lemma~\ref{lem:P1ABDcoNPhard}. Third, assume that $C \subseteq \clos{\Gamma} \subseteq D$ for $C \in \{\IS{2}{1}, \IM, \IR_1, \ID, \IL\}$ and $D \in \{\IE_2, \IL_2, \IV_2, \ID_2\}$. Then $\SAT(\Gamma^+) \in \Ptime$~\cite{Schaefer78} 
    and we therefore (via Lemma~\ref{lemma:sat_to_rep}) conclude that $\RepABD(\Gamma)$ is in $\co\NP$. For hardness, first assume that $\clos{\Gamma} \neq \IR_1$. Then we similarly to the above case observe that $\Gamma$ must contain a non-constant relation, and we apply Lemma~\ref{lem:P1ABDcoNPhard} for the desired result. For $\clos{\Gamma} = \IR_1$ we observe that any pp-definition of $T$ (possibly using $R_=$ and $\mathsf{t}$) can be simplified into an equivalent pp-definition still defining $T$, and we then apply Lemma~\ref{lem:P1ABDcoNPhard}. Fourth, the case when $\IBF \subseteq \clos{\Gamma} \subseteq \IR_0$ and $\mathsf{t} \in \closneq{\Gamma}$ also follows from Lemma~\ref{lem:P1ABDcoNPhard}.

    It can be verified that the only remaining case is when $\mathsf{t} \notin \closneq{\Gamma}$ and if $\clos{\Gamma} = \IBF$ or $\clos{\Gamma} = \IR_0$. We assume that $\clos{\Gamma} = \IR_0$ since it subsumes the other case.  We now apply Lemma~\ref{lemma:weak_def} and conclude that $\closneq{\Gamma} = \closneq{\{\emptyset, F\}}$. In this case $\RepABD(\Gamma)$ can be solved in polynomial time. For an instance $(\KB, H, M, S, k)$ we first check if $\KB = \emptyset$. If so, $H = M = \emptyset$, and the only possible explanation is $E = \emptyset$, and we answer yes if $S = \{\emptyset\}$ and no if $S = \emptyset$. If, on the other hand, $\KB \neq \emptyset$ then we first check if there exists a constraint $\emptyset(x)$ in $\KB$. If so, then $\AllExpl((\KB, H, M)) = \emptyset$ since $\KB$ is not satisfiable, and we can answer yes or no via a simple case analysis. Otherwise, we must have $F(x)$ for each atom in $\KB$, but then we either have no explanation, or $E = \emptyset = M$  as the only possible explanation, and we can easily answer yes or no.
    \end{proof}


We continue by considering the subset-minimal variant $\RepABDsub$  of $\RepABD$, which, surprisingly, turns out to be easier for certain fragments.
We obtain tractability for essentially negative and essentially positive fragments, as long as the equality constraint cannot be expressed. Recall that $\text{EP}^-$ and $\text{EN}^-$ denote the sets of strictly essentially positive, respectively negative, clauses. 



\begin{lemma}\label{lem:RepABDsub-inP}
    $\RepABDsub(\Gamma)\in\Ptime$ if $\Gamma \subseteq \closneq{\mathrm{EP}^-}$ or $\Gamma \subseteq \closneq{\mathrm{EN}^-}$.
\end{lemma}

\begin{proof}
Let $(\KB, H, M, S, k)$ be an instance of $\RepABDsub(\Gamma)$. We assume that each $R \in \Gamma$ is represented by a conjunctive formula over $\text{EP}^-$  (or $\text{EN}^-$). We can then without loss of generality assume that each atom in $\KB$ is from $\text{EP}^-$ (respectively, from $\text{EN}^-$).
    
    In EN$^-$, variables in $M$ can either appear as positive singletons, or as negative literals in a negative clause. If the abduction problem has an explanation, then there exists a unique subset-minimal explanation 
    $E = \{M \cap H\}\setminus \{m \mid (m) \in \KB\}$. This is because every $m \in M$ is either in $H$ and entails itself, or $m$ is already true by being a positive singleton, or there is no solution.

 Similarly, in EP$^-$, any $m \in M$ is either a singleton that is already always true, or $m \in H$ entails itself, or the problem has no solution. Again, there can only be a single subset-minimal explanation $E = \{M \cap H\}\setminus \{m \mid (m) \in \KB\}$.

It suffices to check if $\forall c_i \in S: d(E,c_i) > k$, which can be done in linear time. This concludes the proof.
\end{proof}

If equality can be expressed, however, we obtain $\co\NP$-hardness.

\begin{lemma} ($\star$) \label{lem:RepABDsub-coNP-hard}
    $\RepABDsub(\Gamma)$ is coNP-hard for any language $\Gamma$ such that $R_= \in \closneq{\Gamma}$.
\end{lemma}

Furthermore, we also obtain all hardness results from $\ABD$ again, analogously to Lemma~\ref{lem:ABD-to-RepABD}.

\begin{lemma} ($\star$) \label{lem:ABD-to-RepABDsub}
Let $\Gamma$ be a constraint language. Then $\overline{\ABD(\Gamma)} \preduction \RepABDsub(\Gamma)$.
\end{lemma}

The following Theorem summarizes the results on $\RepABDsub$.

\begin{theorem}($\star$)
Let $\Gamma$ be a constraint language. Then $\RepABDsub(\Gamma)$ is
\begin{enumerate}
    \item $\Pi_2^P$-complete if $\IN_2 \subseteq \clos{\Gamma} \subseteq \II_2$ or $\II_0 \subseteq \clos{\Gamma} \subseteq \II_2$.
    \item $\co\NP$-hard and $\NP$-hard if $\IN \subseteq \clos{\Gamma}$.
    \item $\co\NP$-complete if $C \subseteq \clos{\Gamma} \subseteq D$ for $C \in \{\IM, \ID, \IL\}$ and $D \in \{\IE_2, \IL_2, \IV_2, \ID_2\}$.
    \item $\co\NP$-complete if $\IBF \subseteq \clos{\Gamma} \subseteq D$ for $D \in \{\IS{}{12}, \IS{}{02}\}$, and $\Gamma\not\subseteq \closneq{\mathrm{EP}^-}$ and $\Gamma\not\subseteq \closneq{\mathrm{EN}^-}$.
    \item $\in \Ptime$ otherwise ($\Gamma\subseteq \closneq{\mathrm{EP}^-}$ or $\Gamma\subseteq \closneq{\mathrm{EN}^-}$).
\end{enumerate}
\end{theorem}

\section{Parameterized complexity} \label{sec:para}

Having exhausted all possible sources of polynomial-time solvability we now turn our attention to parameterized complexity. We write $\pRepABD(\Gamma, \ell)$ where $\ell$ is the type of parameter in question. We consider several different parameters beginning with $k$ itself (i.e., the $k$ in the given $\RepABD(\Gamma)$ instance), and then continuing with $|H|$, $|M|$, and $|S|$. 




\subsection{Parameter $k$}

Given a $\RepABD(\Gamma)$ instance $(\KB, H, M, S, k)$ we first consider the $k$ itself as parameter. While this may feel like an obvious parameter choice we will soon prove that the parameter is not likely to help much (coW[1]-hard) for most choices of $\Gamma$.

\begin{lemma} 
$\pRepABD(\Gamma,k)$ is coW[1]-hard for any $\Gamma$ such that $\IS{2}{1} \subseteq \clos{\Gamma}$.
\end{lemma}

\begin{proof}
    We begin by giving a reduction from the W[1]-complete problem $\wsat$ to the complement of $\pRepABD(\Delta, k)$, where $\Delta = \{(\neg x_1 \lor \neg x_2)\}$ contains a single negative 2-clause. In the end, we show why this also gives coW[1]-hardness for $\pRepABD(\Gamma,k)$. 
    
    Hence, let $(\varphi, k)$ be an instance of $\wsat$, that is, $\varphi$ is a negative $2\text{-}CNF$ formula over $n$ variables, $x_1, \dots, x_n$,  and the question is whether there is a model of weight $\geq k$. We map $(\varphi, k)$ to the instance $(\KB, H, M, S, k')$ of $\pRepABD(\Delta, k)$, where
        $\KB = \varphi$, $H  = \{x_1, \dots, x_n\}$, $M = \emptyset$, $S = \{\emptyset\}$, $k' = k + 1$.
    Note that $\KB$ only uses the constraint $(\neg x \lor \neg y)\in\IS{2}{1}$.
    
    To prove correctness we use the following observation.
    \begin{equation}\label{obsRep}
    \bigcup_{E \in S} S_{k'}(E) = S_{k'}(\emptyset) =\text{all explanations of weight} \leq k' < k
    \end{equation}
    In other words, $S$ represents precisely all explanations of weight less than $k$.

    Assume $(\varphi,k)$ is a positive instance of $\wsat$. That is, there is a model $\sigma$ of $\varphi$ of weight at least $k$. We define $E = \{x \in H \mid \sigma(x) = 1\}$ and note that $|E| \geq k$. By construction, $\KB \wedge E$ is consistent and entails $M = \emptyset$. Therefore, $E$ is an explanation. Since $|E| \geq k$, $E$ is not represented by $S$ (confer observation~\ref{obsRep}). Therefore, $(\KB, H, M, S, k')$ is a negative instance.

    Conversely, assume $(\KB, H, M, S, k')$ is a negative instance. That is, there is an explanation $E$ that is not represented. By observation~\ref{obsRep} we conclude that $|E| \geq k$. Since $E$ is an explanation, $\KB \wedge E$ must be consistent. We conclude that $\varphi$ must admit a model of weight $\geq k$, thus $(\varphi, k)$ is a positive instance of  $\wsat$.

    Last, we observe that $\Delta$ contains a single prime relation (no fictitious or redundant arguments). Hence, if $\Delta \subseteq \clos{\Gamma}$ then $\Delta \subseteq \closneq{\Gamma}$ (Lemma~\ref{lemma:can_define}), and we then apply the reduction in Lemma~\ref{lem:baseIndneq} together with the observation that this reduction does not affect the parameter $k$ at all.
\end{proof}

We prove an analogous bound for any language $\Gamma$ that can express implication.

\begin{lemma} ($\star$)
$\pRepABD(\Gamma,k)$ is coW[1]-hard for any $\Gamma$ such that $\IM \subseteq \clos{\Gamma}$.
\end{lemma}

The three main classes that we are missing for a complete classification are linear equations ($\Gamma \subseteq \IL_2$), complementive languages ($\Gamma \subseteq \IN_2$), and any language below $\IS{}{02}$ (essentially positive). We do not fully manage to describe these cases but can for the latter at least prove that its parameterized complexity essentially coincides with the parameterized complexity of the covering radius problem (with parameter $r$). We prove a slightly stronger result and prove that we only need to consider strictly essentially positive clauses, i.e., we do not need the equality relation in the reduction.

\begin{lemma}  ($\star$)\label{lemma:covering_radius_equivalence}
The covering radius problem (with parameter $r$) is FPT-equivalent with the complement of $\pRepABD(\mathrm{EP}^-, k)$.
\end{lemma}

\subsection{Parameter $|H|$}

Next, we consider the size of the hypothesis, $|H|$, as parameter, where we obtain a straightforward parameterized dichotomy. The FPT case can be proven as follows.

\begin{lemma}
$\pRepABD(\Gamma,|H|)\in\FPT$ if $\Gamma$ is Schaefer.
\end{lemma}
\begin{proof}
    We can loop over all explanation candidates $E \subseteq H$ in brute force time $2^{|H|}$. Since the parameter is $|H|$, this is $\FPT$-time. Now we only need to check for each such candidate $E$ 1) whether it is an explanation, and 2) if so, whether it is $k$-represented by $S$ or not. Step 1) amounts to checking whether $\KB \land E$ is satisfiable and whether $\KB \land E \models M$. Both checks can be achieved in polynomial time since for Schaefer languages $\SAT(\Gamma^+)\in \Ptime$. Step 2) amounts to looping over the elements of $S$ and computing for each $E'\in S$ whether $d(E,E') \leq k$ or not. This can be done in polynomial time, since $S$ is part of the input.
\end{proof}

\ifshort
Using two hardness results from~\cite{MahmoodEtAl21} this leads to the following dichotomy (see the full version for details).
\fi

\begin{theorem}
Let $\Gamma$ be a constraint language. Then $\pRepABD(\Gamma,|H|)$ is
\begin{enumerate}
    \item $\para\text{-}\co\DP$-hard if $\IN_2 \subseteq \clos{\Gamma}$ or $\II_0 \subseteq \clos{\Gamma}$.
    \item $\para\text{-}\NP$-hard if $\IN \subseteq \clos{\Gamma}$.
    \item $\in \FPT$ otherwise (that is, $\Gamma \subseteq D$ for $D \in \{\IE_2, \IL_2, \IV_2, \ID_2\}$).
\end{enumerate}
\end{theorem}

\subsection{Parameter $|M|$}

For the size of the manifestation ($|M|$) we currently lack FPT cases, but, on the other hand, can show para-coNP-hardness for the majority of constraint languages.

\begin{lemma}
        $\pRepABD(\Gamma,|M|)$ is para-coNP-hard for any $\Gamma$ such that $\mathsf{t} \in \closneq{\Gamma}$ or $T \in \closneq{\Gamma}$.
\end{lemma}

\begin{proof}
    The construction is the same as in Lemma~\ref{lem:P1ABDcoNPhard}. We note that in that construction $|M|$ is constant.
\end{proof}

\subsection{Parameter $|S|$}

The last parameter that we consider is the number of sets $|S|$. We first observe a hard case by using Lemma~\ref{lem:ABD-to-RepABD}.

\begin{lemma} $(\star)$
$\pRepABD(\Gamma, |S|)$ is para-coNP-hard for any $\Gamma$ such that $\IS{2}{11} \subseteq \clos{\Gamma}$.
\end{lemma}



\begin{figure*}[htp]
    \centering
    \scalebox{0.5}{%
        \includegraphics{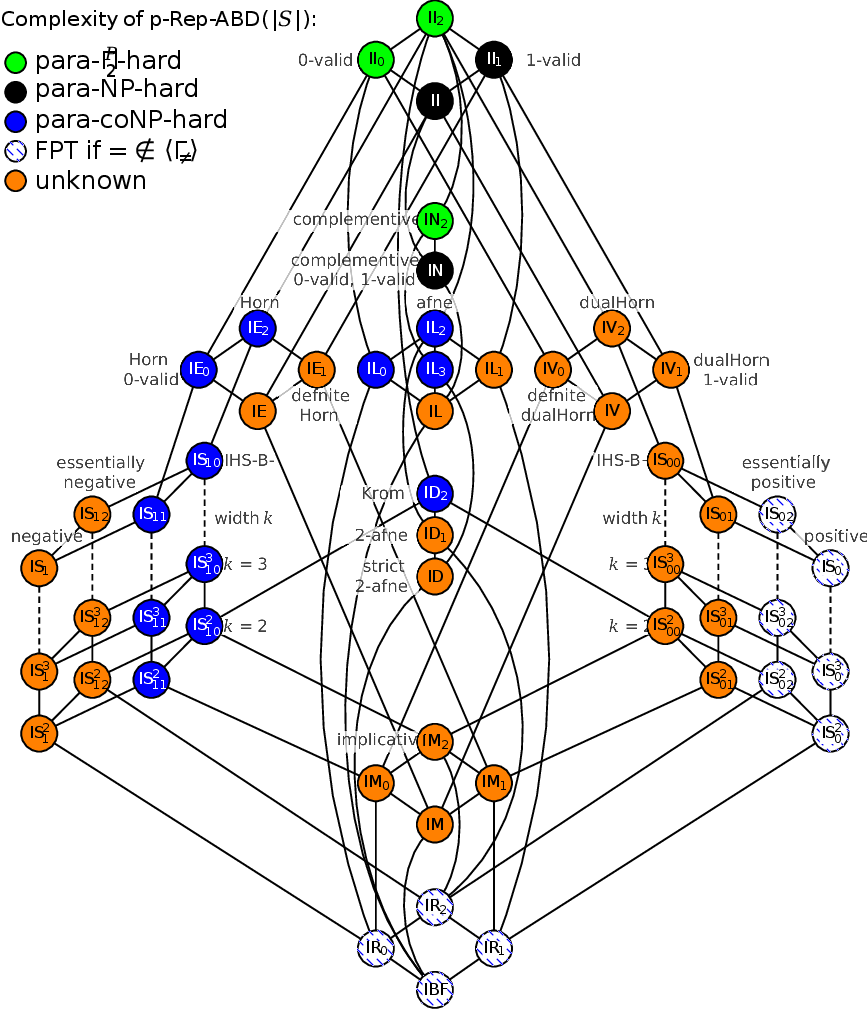}%
    }
    \caption{The complexity of $\pRepABD(\Gamma, |S|)$ illustrated via Post's lattice.}
    \label{fig:post-latticepABD-S}
\end{figure*}

\begin{lemma}
$\pRepABD(\mathrm{EP}^-, |S|)$ is in FPT.
\end{lemma}

\begin{proof}
We first observe that $\pRepABD(\mathrm{EP}^-, |S|)$ is FPT-equivalent to the covering radius problem by the same reduction used in Lemma~\ref{lemma:covering_radius_equivalence}. 

Next, we show FPT for this problem by giving an FPT-reduction to the \emph{closest string problem} with alphabet $\Sigma = \{0,1\}$. Here, we are given a set of strings $C$, a radius $r$ and want to know if there exists $x \in \{0,1\}^n$ such that $d(x,c_i) \leq r$ for every $c_i \in C$. This problem is known to be FPT with respect to $|C|$~\cite{GrammNR03}. Given an instance $C \subseteq \{0,1\}^n$ and $r \geq 1$ of the covering radius problem we map it to an instance $(C', r')$ of the closest string problem as follows.  
\begin{itemize}
    \item $C' = \{\overline{c_i} \mid c_i \in C\}$ where $\overline{c_i}$ is the bit-wise complement of $c_i$,
    \item $r' = n-r$.
\end{itemize}

The correctness is as follows: assume $\exists x \in\{0,1\}^*: \forall c_i \in C: d(x,c_i) > r \leftrightarrow |x \cap c_i| > r \leftrightarrow |x \cap \overline{c_i}| \leq n-r$. Moreover, the reduction can clearly be carried out in FPT time (indeed, even polynomial time) and preserves the parameter since $|C'| = |C|$. This concludes the proof.
\end{proof}

This result extends to $\pRepABD(\Gamma, |S|)$ for any $\Gamma \subseteq \closneq{\mathrm{EP}^-}$ since the reduction in Lemma~\ref{lem:baseIndneq} does not affect the parameter $|S|$.




\section{Conclusion} \label{sec:conclusions}

\ifshort
We introduced representative sets in propositional abduction
and illustrated how it  gives a better understanding of the space of possible explanations; for example, (recall Section~\ref{sec:intro}), in a medical situation, representative abduction can tell doctors if their set of hypothesis has covered all major possibilities to explain the symptoms or if they might have missed something crucial.
Complexity-wise, we have established an almost complete classification of the problem for constraint language restrictions (Post's lattice). 
Although the problem admits almost no tractable cases in its normal form, the subset minimal variant and the parametrized complexity approach give us several additional tractable cases. Let us now discuss some potential future research directions. 
\fi 




\paragraph{Completing the classification}
We have three open cases shared by representative abduction and its subset minimal variant. These are the 1-valid languages. We have already established that they are both NP- and coNP-hard. It would be interesting to investigate their exact complexity class, and a good candidate for these languages is the class $\DP$. The techniques that would be used to prove it would most likely be novel for abduction-related problems.

\paragraph{Parametrized complexity}
We have investigated the parametrized complexity of representative abduction with different parameters, namely $k$ (the Hamming distance), $|H|$, $|M|$, and $|S|$ (size of the representative set). For parameter $|H|$ we have established a complete classification, obtained $\FPT$-results for a large portion of languages (Schaefer languages), and proven hardness for the rest. We have established a para-coNP-hardness result for the majority of languages for parameter $|M|$. In our opinion, the parameters $k$ and $|S|$ are perhaps the most interesting ones. For parameter $|S|$ we have established both hardness for a large portion of the lattice, but also non-trivial $\FPT$ results for $\epos^-$-languages, but which still leaves some interesting open cases for complementive and affine languages. Finally, for parameter $k$ we have established non-trivial coW[1]-hardness results for languages $\IM$ and $\eneg^-$, leaving what are perhaps the most interesting open cases for future research. The parametrized complexity approach seems to us to be extremely interesting. It establishes possibly useful FPT results, as well as interesting reductions for proving hardness, and it links our problem to know problems from coding theory.

\paragraph{Relation to coding theory}
We have established a strong connection between our representative abduction problem and the covering radius problem from coding theory, that as far as we know has not been established before. 
A lot of classical and parametrized complexity results from coding theory  were instrumental to our lattice classifications, such as NP-completeness of the covering radius problem~\cite{frances1997covering} and the FPT result for the closest string problem with parameter $|C|$~\cite{GrammNR03}. Future research on parametrized complexity with parameter $k$ for representative abduction and parameter $r$ (radius) for the covering radius problem, will greatly benefit both fields.






\bibliographystyle{eptcs}
\bibliography{ref}

\end{document}